\documentclass[reqno]{amsart}
\usepackage[utf8]{inputenc}
\usepackage{pmboxdraw}
\usepackage{amsmath}
\usepackage{amssymb}
\usepackage{xcolor}
\usepackage{graphicx}
\usepackage{float}
\usepackage[normalem]{ulem}
\usepackage[hidelinks]{hyperref}
\usepackage{cite}
\newtheorem{theorem}{Theorem}
\newtheorem{remark}{Remark}

\newtheorem{proposition}[theorem]{Proposition}
\newtheorem{conjecture}{Conjecture}
\newtheorem{corollary}[theorem]{Corollary}
\newtheorem*{mainresult}{Main result}

\calclayout

\newcommand{\C}{\mathbb{C}}
\newcommand{\Z}{\mathbb{Z}}
\newcommand{\R}{\mathbb{R}}

\newcommand{\ler}[1]{\left( #1 \right)}
\newcommand{\lesq}[1]{\left[ #1 \right]}
\newcommand{\lers}[1]{\left\{ #1 \right\}}
\newcommand{\hohc}{\cH \otimes \cH^*}

\newcommand{\abs}[1]{\left| #1 \right|}

\newcommand{\norm}[1]{\left|\left|#1\right|\right|}
\newcommand{\bra}[1]{\langle #1 |}
\newcommand{\ket}[1]{| #1 \rangle}
\newcommand{\Bra}[1]{\langle\langle #1 ||}
\newcommand{\Ket}[1]{|| #1 \rangle\rangle}

\newcommand{\be}{\begin{equation}}
\newcommand{\ee}{\end{equation}}
\newcommand{\ba}{\begin{array}}
\newcommand{\ea}{\end{array}}

\newcommand{\fel}{\frac{1}{2}}
\newcommand{\cB}{\mathcal{B}}
\newcommand{\cH}{\mathcal{H}}

\newcommand{\cP}{\mathcal{P}}
\newcommand{\cT}{\mathcal{T}}

\newcommand{\cS}{\mathcal{S}}
\newcommand{\cC}{\mathcal{C}}

\newcommand{\cW}{\mathcal{W}}
\newcommand{\cL}{\mathcal{L}}
\newcommand{\bb}{\mathbf{b}}

\newcommand{\tr}{\mathrm{tr}}
\newcommand{\dd}{\mathrm{d}}

\newcommand{\cba}{\mathcal{A}}
\newcommand{\cA}{\mathcal{A}}

\newcommand{\dwts}[2]{d_{\cW_2}^2(#1,#2)}
\newcommand{\E}{\mathbb{E}}
\newcommand{\law}{\mathrm{law}}

\newcommand\lh{\cL(\cH)}

\newcommand{\sh}{\cS\ler{\cH}}

\title[On the metric property of quantum Wasserstein divergences]{On the metric property of quantum Wasserstein divergences}

\author[Gergely Bunth]{Gergely Bunth}
\address{Gergely Bunth, HUN-REN Alfr\'ed R\'enyi Institute of Mathematics\\ Re\'altanoda u. 13-15.\\Budapest H-1053\\ Hungary}
\email{bunth.gergely@renyi.hu}

\author[J\'ozsef Pitrik]{J\'ozsef Pitrik}
\address{J\'ozsef Pitrik, HUN-REN Wigner Research Centre for Physics\\ Budapest H-1525, Hungary\\ and HUN-REN Alfr\'ed R\'enyi Institute of Mathematics\\ Re\'altanoda u. 13-15.\\ Budapest H-1053\\ Hungary\\ and Department of Analysis, Institute of Mathematics \\Budapest University of Technology and Economics\\ M\H{u}egyetem rkp. 3. \\ Budapest H-1111\\ Hungary}
\email{pitrik.jozsef@renyi.hu}

\author[Tam\'as Titkos]{Tam\'as Titkos}
\address{Tam\'as Titkos, Corvinus University of Budapest\\ Department of Mathematics\\ Fővám tér 13-15.\\ Budapest H-1093\\Hungary\\ 
and \\ HUN-REN Alfr\'ed R\'enyi Institute of Mathematics\\ Re\'altanoda u. 13-15.\\ Budapest H-1053\\Hungary}
\email{titkos.tamas@renyi.hu}

\author[D\'aniel Virosztek]{D\'aniel Virosztek}
\address{D\'aniel Virosztek, HUN-REN Alfr\'ed R\'enyi Institute of Mathematics\\ Re\'altanoda u. 13-15.\\Budapest H-1053\\ Hungary}
\email{virosztek.daniel@renyi.hu}

\date{}



\thanks{Bunth was supported by the Momentum Program of the Hungarian Academy of Sciences (grant no. LP2021-15/2021); Pitrik was supported by the “Frontline” Research Excellence Programme of the Hungarian National Research, Development and Innovation Office - NKFIH (grant no. KKP133827) and by the Momentum Program of the Hungarian Academy of Sciences (grant no. LP2021-15/2021); Titkos was supported by the Hungarian National Research, Development and Innovation Office - NKFIH (grant no. K115383) and by the Momentum Program of the Hungarian Academy of Sciences (grant no. LP2021-15/2021).; Virosztek was supported by the Momentum Program of the Hungarian Academy of Sciences (grant no. LP2021-15/2021) and by the ERC Consolidator Grant no. 772466.}

\begin{document}

\maketitle

\begin{abstract}
Quantum Wasserstein divergences are modified versions of quantum Wasserstein distances defined by channels, and they are conjectured to be genuine metrics on quantum state spaces by De Palma and Trevisan. We prove triangle inequality for quantum Wasserstein divergences for every quantum system described by a separable Hilbert space and any quadratic cost operator under the assumption that a particular state involved is pure, and all the states have finite energy. We also provide strong numerical evidence suggesting that the triangle inequality holds in general, for an arbitrary choice of states.
\end{abstract}

\tableofcontents

\section{Introduction}\label{s:intro}

\subsection{Motivation and main result}\label{ss:mot-res}

Although the classical problem of transporting mass in an optimal way was formulated in the 18th century by Monge, the theory of classical optimal transport (OT) became one of the central topics of analysis only in the last few decades with intimate links to mathematical physics \cite{DL-DP-M-T,JKO-98,Lott}, PDE theory \cite{Ambrosio-lecturenotes,Figalli-note,OPS}, and probability \cite{BP,BiegelbockJuillet,Hairer2}. Not to mention the countless applications in artificial intelligence, image processing, and many other fields of applied sciences. See e.g. \cite{Peyre1,Peyre2,Santambrogio} and the references therein.
\par
The quantum counterpart of classical OT is just emerging in these years. As always, the correspondence between the classical world and the quantum world is not one-to-one. Non-commutative optimal transport is a flourishing research field these days with several different promising approaches such as that of Biane and Voiculescu \cite{BianeVoiculescu}, Carlen, Maas, Datta, and Rouzé  \cite{CarlenMaas-1,CarlenMaas-2,CarlenMaas-3,DattaRouze1,DattaRouze2}, Caglioti, Golse, Mouhot, and Paul \cite{CagliotiGolsePaul, CagliotiGolsePaul-towardsqot, GolseMouhotPaul, GolsePaul-wavepackets, GolsePaul-Schrodinger,GolsePaul-meanfieldlimit}, De Palma and Trevisan \cite{DePalmaMarvianTrevisanLloyd,DPT-AHP},  \.Zyczkowski and his collaborators  \cite{FriedmanEcksteinColeZyczkowski-MK,ZyczkowskiSlomczynski1,ZyczkowskiSlomczynski2,BistronEcksteinZyczkowski}, and Duvenhage \cite{Duvenhage1, Duvenhage2}. Separable quantum Wasserstein distances have also been introduced recently \cite{TothPitrik}.
From our viewpoint, the most relevant approach is the one of De Palma and Trevisan involving quantum channels, which is closely related to the quantum optimal transport concept of Caglioti, Golse, Mouhot, and Paul based on quantum couplings.
\par
It is a common feature of both the channel-based and the coupling-based quantum optimal transport distances that \emph{they are not genuine metrics} --- in particular, states may have a positive distance from themselves. This phenomenon is natural on the one hand, considering the nature of quantum mechanics, but has counter-intuitive consequences on the other hand. One of these consequences is that there exist non-surjective and even non-injective quantum Wasserstein isometries (i.e., distance preserving maps) on the compact space of states of a finite-level quantum system \cite{GPTV23} --- none of these could possibly happen in a genuine metric setting. 
\par
As a response to this phenomenon, De Palma and Trevisan introduced \emph{quantum Wasserstein divergences} \cite{DPT-lecture-notes} which are appropriately modified quantum Wasserstein distances --- see \eqref{eq:mod-qw-dist-def} for a precise definition. They conjectured that quantum Wasserstein divergences are genuine metrics on quantum state spaces. This paper is devoted to the question whether the triangle inequality holds for these divergences. We formulate our main result in an informal way below --- see Theorem \ref{thm:main} for the precise statement.

\begin{mainresult}
For every quantum system described by a separable Hilbert space $\cH,$ and for every finite collection $\cA$ of observable quantities, the corresponding quadratic quantum Wasserstein divergence $d_{\cA}$ satisfies
\be \label{eq:tri-in-plain}
d_{\cA}(\rho, \tau) \leq d_{\cA}(\rho, \omega)+d_{\cA}(\omega, \tau)
\ee
for any triplet $(\rho, \omega, \tau)$ of states assuming that $\omega$ is pure or both $\rho$ and $\tau$ are pure, and all the states involved are of finite energy. Moreover, numerical results strongly indicate that \eqref{eq:tri-in-plain} holds for any triplet $(\rho, \omega, \tau)$ of states without any further assumptions.
\end{mainresult}

This paper is organized as follows. In the next subsection (Subsec. \ref{ss:not-not}) we introduce all the necessary notions and notation. Section \ref{s:main-proof} is dedicated to the analytical proof of Theorem \ref{thm:main}, which is a precise formulation of the above main result. In Section \ref{s:numerics} we present numerical evidence suggesting that the conclusion of Theorem \ref{thm:main} holds without any extra assumptions on the states involved. Encouraged by the numerics presented in Section \ref{s:numerics}, we make steps in the direction of an analytic proof of the triangle inequality in full generality --- this material is presented in Section \ref{s:anal-spec}.
\par
Finally, we remark the analogy of our present problem with the problem of the metric property of the \emph{quantum Jensen-Shannon divergence}. In that case, the triangle inequality was proved in 2008 for pure states analytically and for mixed states numerically \cite{lmbcp-pra-77-2008} (see also \cite{briet-harremoes}). A decade later, the problem was settled by an analytical proof in full generality --- see \cite{V2021metric-JS} and \cite{sra-laa}.

\subsection{Basic notions, notation}\label{ss:not-not}

The {\it classical OT problem} is to arrange the transportation of goods from producers to consumers in an optimal way, given the distribution of production and consumption (described by probability measures $\mu$ and $\nu$), and the cost $c(x,y)$ of transporting a unit of goods from $x$ to $y.$ Accordingly, a {\it transport plan} is modeled by a probability distribution $\pi$ on the product of the initial and the target space, where $\dd \pi(x,y)$ is the amount of goods to be transferred from $x$ to $y,$ and hence the marginals of $\pi$ are $\mu$ and $\nu.$ So the \emph{optimal transport cost} is the infimum of a convex optimization problem with linear loss function:
\be \label{eq:class-ot-cost-def}
\mathrm{Cost}\ler{\mu,\nu,c}= \inf \lers{ \iint_{X \times Y} c(x,y) \dd \pi\ler{x,y} \, \middle| \, (\pi)_1=\mu, \, (\pi)_2=\nu}
\ee
where $(\pi)_i$ denotes the $i$th marginal of $\pi,$ and $X$ is the initial and $Y$ is the target space. Under mild continuity assumptions on $c(x,y)$ --- lower semi-continuity does the job --- the infimum in \eqref{eq:class-ot-cost-def} is actually a minimum as it is realized by a transport plan. These minimizing plans are called \emph{optimal transport plans.}
\par
One has great freedom in choosing the cost function $c(x,y).$ However, the \emph{quadratic cost} $c(x,y)=r^2(x,y),$ which is simply the square of the distance, plays a distinguished role. Its importance comes mainly from fluid dynamics and the dynamical theory of optimal transportation --- the exponent $p=2$ is distinguished by the fact that the kinetic energy is proportional to the \emph{square} of the velocity.
Accordingly, the relevance of quadratic Wasserstein spaces has grown dramatically in recent decades due to their close connection with PDE theory and gradient flows. Recall that if $(X,r)$ is complete and separable metric space, then the classical quadratic Wasserstein space $\mathcal{W}_2(X)$ is the collection of those probability measures on the Borel $\sigma$-algebra $\mathcal{B}$ that satisfy 
$\int_X r(x,x_0)^2~\mathrm{d}\mu(x)<\infty$ for some $x_0\in X$,
endowed with the quadratic Wasserstein distance
\begin{equation} \label{eq:dw2-def}
\dwts{\mu}{\nu}:=\inf_{\pi} \int_{X \times X} r^2(x,y)~\mathrm{d} \pi(x,y)
\end{equation}
where the infimum runs over all couplings of $\mu$ and $\nu.$
\par

In classical mechanics, the state of a particle moving in $\R^d$ is described by a probability measure $\mu$ on the phase space $\R^d \times \R^d$ which is the collection of all possible values of the position and momentum variables $q, p \in \R^d.$ In this concrete setting, the quadratic Wasserstein distance \eqref{eq:dw2-def} of the classical states $\mu, \nu \in \cP\ler{\R^{2d}}$ is given by
\be \label{eq:dw2-cl-mech}
\dwts{\mu}{\nu}=\inf_{\pi} \lers{ \iint_{\R^{2d}\times \R^{2d}} \abs{(q_1,p_1)-(q_2,p_2)}^2 \dd \pi\ler{(q_1,p_1),(q_2,p_2)}}
\ee
where $\pi \in \cP(\R^{2d} \times \R^{2d})$ and $(\pi)_1=\mu, \, (\pi)_2=\nu$ --- here $(\pi)_i$ denotes the $i$th marginal of $\pi.$
Recall that Wasserstein distances admit a picturesque probabilistic interpretation as they are defined by optimization over couplings of probability measures. Let us stick to the concrete case of $2$-Wasserstein distances between states of classical mechanical systems. In this case, the probabilistic version of \eqref{eq:dw2-cl-mech} reads as follows:
\be \label{eq:def-clm}
\dwts{\mu}{\nu}
=\inf \lers{ \E \abs{(Q_1,P_1)-(Q_2,P_2)}^2 \, \middle| \, \law(Q_1,P_1)=\mu, \, \law(Q_2,P_2)=\nu}.
\ee
In the above formula, the random variables $Q_j$ and $P_j$ represent the position and the momentum of the $j$th particle ($j=1,2$). That is, we minimize the sum of the expected squared differences between the positions of the two particles and the momenta of them.
\par

In quantum mechanics, the state of a particle moving in $\R^d$ is described by a wave function $\psi \in L^2\ler{\R^d}$ of unit norm, or more generally, by a normalized, positive, trace-class operator $\rho$ on $\cH=L^2\ler{\R^d}.$ Measurable physical quantities correspond to (possibly unbounded) self-adjoint operators on $\cH=L^2\ler{\R^d}.$ The spectrum of such an operator is precisely the collection of all possible outcomes of a quantum measurement. In the sequel, we denote by $\lh^{sa}$ the set of self-adjoint but not necessarily bounded operators on $\cH$, and $\sh$ stands for the set of states, that is, the set of positive trace-class operators on $\cH$ with unit trace. The space of all bounded operators on $\cH$ is denoted by $\cB(\cH),$ and we recall that the collection of trace-class operators on $\cH$ is denoted by $\cT_1(\cH)$ and defied by $\cT_1(\cH)= \lers{X \in \cB(\cH) \, \middle| \, \tr_{\cH}[\sqrt{X^*X}] < \infty}.$ Similarly, $\cT_2(\cH)$ stands for the set of Hilbert-Schmidt operators defined by $\cT_2(\cH)= \lers{X \in \cB(\cH) \, \middle| \, \tr_{\cH}[X^*X] < \infty}.$ When measuring an observable quantity $A \in \lh^{sa}$ on a quantum system being in the state $\rho \in \sh,$ the probability of the outcome lying in an interval $[a,b]\subset \R$ is $\tr_{\cH} \ler{\rho E_A\ler{[a,b]}},$ where $E_A$ is the spectral measure of $A.$
Consequently, a quantum state encapsulates a bunch of classical probability distributions, each corresponding to a physical quantity we are interested in. 
\par
Let us single out a few observable quantities $A_1, \dots, A_k \in \lh^{sa}$ we are interested in, and let $X_{j}^{(\rho)}$ denote the random variable obtained by measuring $A_j$ in the initial state $\rho,$ and let $X_{j}^{(\omega)}$ denote the random variable obtained by measuring $A_j$ in the final state $\omega.$ According to \eqref{eq:def-clm}, the squared OT distance of the quantum states $\rho, \omega \in \sh$ should read as
\be \label{eq:def-qot-1}
D^2\ler{\rho, \omega}=\inf\lers{\sum_{j=1}^k \E \ler{X_{j}^{(\rho)}-X_{j}^{(\omega)}}^2 }
\ee
where the infimum is taken over all possible couplings of the quantum states $\rho$ and $\omega.$ According to the convention introduced by De Palma and Trevisan \cite{DPT-AHP}, the set of all couplings of the quantum states $\rho, \omega \in \cS\ler{\cH}$ is denoted by $\cC\ler{\rho, \omega},$ and is given by
\be \label{eq:q-coup-def}
\cC\ler{\rho, \omega}=\lers{\Pi \in \cS\ler{\cH \otimes \cH^*} \, \middle| \, \tr_{\cH^*} [\Pi]=\omega, \,  \tr_{\cH} [\Pi]=\rho^T},
\ee
where the \emph{transpose} $A^T$ of a linear operator $A$ acting on $\cH$ is a linear operator on the dual space $\cH^*$ defined by the identity $(A^T \eta) (\varphi) \equiv \eta (A \varphi)$ where $\eta \in \cH^*$ and $\varphi \in \mathrm{dom}(A).$
That is, a coupling of $\rho$ and $\omega$ is a state $\Pi$ on $\hohc$ such that 
$$ 
\tr_{\hohc}[\ler{A\otimes I_{\cH^*}} \Pi]=\tr_{\cH} [\omega A]
$$
and
\be \label{eq:part-trace-def}
\tr_{\hohc}\lesq{\ler{I_{\cH} \otimes B^{T}} \Pi}=\tr_{\cH^*} [\rho^T B^T]=\tr_{\cH} [\rho B]
\ee
for all bounded $A, B \in \cL(\cH)^{sa}.$ Note the clear analogy of the above definition of quantum couplings with the definition of classical couplings that can be rephrased as follows: $\pi \in \cP(X^2)$ is a coupling of $\mu \in \cP(X)$ and $\nu \in \cP(X)$ if 
$$
\iint_{X^2} f(x) \dd \pi (x,y)=\int_{X} f(x) \dd \mu(x) \text{ and } \iint_{X^2} g(y) \dd \pi (x,y)=\int_{X} g(y) \dd \nu(y)
$$
for every continuous and bounded function $f,g$ defined on $X.$
We remark that $\cC\ler{\rho,\omega}$ is never empty, because the trivial coupling $\omega\otimes\rho^T$ belongs to $\cC\ler{\rho,\omega}$.
\par
Note that the definition of couplings \eqref{eq:q-coup-def} proposed by De Palma and Trevisan \cite{DPT-AHP} is different from the definition proposed by Golse, Mouhot, Paul \cite{GolseMouhotPaul} in the sense that it involves the dual Hilbert space $\cH^*$ and hence the transpose operation. For a clarification of this difference, see Remark 1 in \cite{DPT-AHP}. For more detail on the latter concept of quantum couplings, the interested reader should consult \cite{CagliotiGolsePaul,CagliotiGolsePaul-towardsqot,GolseMouhotPaul,GolsePaul-Schrodinger,GolsePaul-Nbody, GolsePaul-OTapproach, GolsePaul-meanfieldlimit}. 
\par
Let us jump back to \eqref{eq:def-qot-1} and note that by Born's rule on quantum measurement, if the state of our composite quantum system is described by $\Pi \in \cC(\rho, \omega),$ then 
\be \label{eq:exp-quant-coup}
\E \ler{X_{j}^{(\rho)}-X_{j}^{(\omega)}}^2
=\tr_{\hohc}\left[ \ler{A_j\otimes I_{\cH^*}-I_{\cH} \otimes A_j^T} \Pi \ler{A_j\otimes I_{\cH^*}-I_{\cH} \otimes A_j^T}\right].
\ee
Here we used the convention that the right-hand side of \eqref{eq:exp-quant-coup} is defined to be $+\infty$ if there is an eigenvector of $\Pi$ outside the domain of $A_j\otimes I_{\cH^*}-I_{\cH} \otimes A_j^T$ for some $j \in \lers{1, \dots, k}.$
Therefore, in view of \eqref{eq:def-qot-1}, \eqref{eq:q-coup-def}, and \eqref{eq:exp-quant-coup}, the quadratic Wasserstein distance of $\rho$ and $\omega$ with respect to the measurable quantities $\lers{A_j}_{j=1}^k=:\cba$ is given by
\be \label{eq:2-Wass-quant-def}
D_{\cba}^2\ler{\rho, \omega}=\inf_{\Pi \in \cC(\rho,\omega)}
\lers{\sum_{j=1}^k
\tr_{\hohc}\left[ \ler{A_j\otimes I_{\cH^*}-I_{\cH} \otimes A_j^T}\Pi \ler{A_j\otimes I_{\cH^*}-I_{\cH} \otimes A_j^T}\right]
}.
\ee
We recall (see \cite[Definition 6]{DPT-AHP}) that the energy of a state $\rho \in \sh$ with respect to the observable $A \in \cL(\cH)^{sa}$ is given by $E_A(\rho)=\sum_{j=1}^{\infty} p_j \norm{A \psi_j}^2,$ where $\sum_{j=1}^{\infty} p_j \ket{\psi_j}\bra{\psi_j}$ is the spectral decomposition of $\rho,$ and $E_A(\rho)=+\infty$ if $\psi_j \notin \mathrm{dom}(A)$ for some $j.$ The energy of $\rho$ with respect to the collection of observables $\cA=\lers{A_1, \dots, A_k}$ is defined by $E_{\cA}(\rho)=\sum_{j=1}^k E_{A_j}(\rho).$ Recall that classical quadratic Wasserstein spaces consist of probability measures with finite second moment. As the natural quantum analogs of them are states with finite energy, we restrict our attention to such quantum states in the sequel.
\par
By \cite[Proposition 3]{DPT-AHP}, if the states $\rho, \omega \in \sh$ have finite energy, then any quantum coupling $\Pi \in \cC(\rho, \omega)$ has finite cost. Moreover, both $A_j \rho A_j$ and $A_j \omega A_j$ are trace-class for every $j \in \lers{1, \dots, k}$ -- see \cite[Lemma 3]{DPT-AHP}. Consequently, by the definition of Hilbert-Schmidt and trace-class operators, $\sqrt{\rho}A_j$ and $\sqrt{\omega}A_j$ are Hilbert-Schmidt, and so are $A_j\sqrt{\rho}$ and $A_j \sqrt{\omega}$ as taking the adjoint is an involution of $\cT_2(\cH).$ Note furthermore that both $\sqrt{\rho}$ and $\sqrt{\omega}$ are Hilbert-Schmidt by definition, and hence the operators $\rho A_j, \, A_j \rho, \, \omega A_j, \, A_j \omega$ are trace-class as they are products of two Hilbert-Schmidt operators.  
\par
A prominent coupling of a state $\rho \in \sh$ with itself is the \emph{canonical purification} $\Ket{\sqrt{\rho}}\Bra{\sqrt{\rho}} \in \cS(\hohc)$ which is the rank-one projection corresponding to the unit vector $\Ket{\sqrt{\rho}} \in \hohc$ obtained from $\sqrt{\rho} \in \mathcal{T}_2(\cH)$ by the canonical isomorphism between $\mathcal{T}_2(\cH)$ and $\hohc.$
\par
An important feature of the quadratic Wasserstein distances is that the distance of a state $\rho$ from itself (which may be positive) is always realized by the canonical purification --- see \cite[Corollary 1]{DPT-AHP}. That is,
$$
D_{\cA}^2\ler{\rho, \rho}
=\sum_{j=1}^k
\tr_{\hohc}\left[ \ler{A_j\otimes I_{\cH^*}-I_{\cH} \otimes A_j^T} \Ket{\sqrt{\rho}}\Bra{\sqrt{\rho}} \ler{A_j\otimes I_{\cH^*}-I_{\cH} \otimes A_j^T}\right]
$$
$$
=\sum_{j=1}^k
\tr_{\hohc}\left[ \Ket{A_j\sqrt{\rho} -\sqrt{\rho}A_j}\Bra{A_j\sqrt{\rho} -\sqrt{\rho}A_j}\right]
=\sum_{j=1}^k \norm{A_j\sqrt{\rho} -\sqrt{\rho}A_j}_{HS}^2
$$
We have seen that the finite energy condition on $\rho$ implies that both $A_j\sqrt{\rho}$ and $\sqrt{\rho}A_j$ are Hilbert-Schmidt operators, and hence not only $A_j \rho A_j$ but also $\sqrt{\rho} A_j \sqrt{\rho} A_j$ and $A_j \sqrt{\rho} A_j \sqrt{\rho}$ and $\sqrt{\rho} A_j^2 \sqrt{\rho}$ are trace-class operators. We note that taking the adjoint leaves both the Hilbert-Schmidt norm and the trace invariant, and hence $\norm{A_j \sqrt{\rho}}_{HS}=\norm{\sqrt{\rho} A_j}_{HS}$ and $\tr_{\cH}[A_j \sqrt{\rho}A_j \sqrt{\rho}]=\tr_{\cH}[\sqrt{\rho}A_j \sqrt{\rho}A_j].$ Consequently,
\be \label{eq:self-dist-explicit}
D_{\cA}^2(\rho, \rho)= \sum_{j=1}^k \tr_{\cH}[2 A_j \rho A_j -2 \sqrt{\rho}A_j \sqrt{\rho}A_j].
\ee
Moreover, the following concavity-like result is true for any choice of $\cA=\lers{A_1, \dots, A_k}$ and for any $\rho, \omega \in \sh$ with finite energy:
\be \label{eq:concavity}
D_{\cA}^2(\rho, \omega) \geq \frac{1}{2}\ler{D_{\cba}^2(\rho, \rho)+D_{\cba}^2(\omega, \omega)}.
\ee
Indeed, \eqref{eq:concavity} is an easy consequence of Theorem 1 and Corollary 1 of \cite{DPT-AHP}. The following conjecture was popularized by Dario Trevisan and Giacomo De Palma in the fall of 2022 --- see also \cite{DPT-lecture-notes}. 

\begin{conjecture}[De Palma-Trevisan, \cite{DPT-lecture-notes}]
\label{conj:truemetric}
A modified version of the quantum optimal transport distance \eqref{eq:2-Wass-quant-def} defined by
\be \label{eq:mod-qw-dist-def}
d_{\cba} (\rho, \omega):=\sqrt{D_{\cba}^2(\rho,\omega)-\frac{1}{2}\ler{D_{\cba}^2(\rho, \rho)+D_{\cba}^2(\omega, \omega)}}
\ee
is a true metric for all finite collections of observables $\cA=\lers{A_1, \dots, A_k}$ on the set of those states on $\cH$ that have finite energy with respect to $\cA$ --- up to some non-degeneracy assumptions on the $A_j$'s to ensure the definiteness of $d_{\cA},$ that is, that $d_{\cA}(\rho, \omega)=0$ only if $\rho=\omega.$
\end{conjecture}

Let us denote by $\cP_1(\cH)$ the set of rank-one ortho-projections on $\cH,$ that is, the set of pure states. Note that if either $\rho$ or $\omega$ is a pure state, then the only coupling of them is the tensor product, that is, $\cC(\rho, \omega)=\lers{\omega\otimes\rho^T}.$ Therefore, by \eqref{eq:2-Wass-quant-def}, the quadratic quantum Wasserstein distance $D_{\cA}(\rho, \omega)$ has the following explicit form in this special case:

$$
D_{\cA}^2(\rho, \omega)
=\sum_{j=1}^k \tr_{\hohc}\left[\ler{A_j\otimes I_{\cH^*}-I_{\cH} \otimes A_j^T}\ler{\omega \otimes \rho^T} \ler{A_j\otimes I_{\cH^*}-I_{\cH} \otimes A_j^T}\right]
$$
$$
=\sum_{j=1}^k \tr_{\hohc}\left[A_j\omega A_j \otimes \rho^T-\omega A_j \otimes A_j^T\rho^T- A_j \omega \otimes \rho^T A_j^T + \omega \otimes A_j^T \rho^T A_j^T \right]
$$
\be \label{eq:qW-dist-expl-pure}
=\tr_{\cH}[A_j\omega A_j]+\tr_{\cH}[A_j\rho A_j]-2 \tr_{\cH}[\omega A_j]\tr_{\cH}[\rho A_j].
\ee
Here we used that $\tr_{\hohc}[X \otimes Y]=\tr_{\cH}[X]\tr_{\cH^*}[Y]$ if $X \in \mathcal{T}_1(\cH)$ and $Y \in \mathcal{T}_1(\cH^*),$ that $\tr_{\cH^*}[X^T]=\tr_{\cH}[X]$ for $X \in \mathcal{T}_1(\cH),$ and that the operators $\rho A_j, \, A_j \rho, \omega A_j, A_j \omega$ are trace-class.

\section{Triangle inequality for quantum Wasserstein divergences --- the proof of Theorem \ref{thm:main}}\label{s:main-proof}

Having introduced all the necessary notions and notation, we are in the position to state and prove our main result.

\begin{theorem} \label{thm:main}
Let $\cH$ be a separable Hilbert space, and let $\cA=\lers{A_j}_{j=1}^k \subset \cL(\cH)^{sa}$ be an arbitrary finite collection of observable quantities, and let $d_{\cba}$ be the corresponding quadratic quantum Wasserstein divergence defined by \eqref{eq:mod-qw-dist-def} and \eqref{eq:2-Wass-quant-def}. Let $\rho, \omega, \tau \in \cS(\cH)$ and assume that $\omega \in \cP_1(\cH)$ or both $\rho$ and $\tau$ are in $\cP_1(\cH).$ Moreover, assume that $\rho, \omega$ and $\tau$ have finite energy with respect to $\cA.$  
Then the triangle inequality
\be \label{eq:tri-in}
d_{\cba}(\rho, \tau) \leq d_{\cba}(\rho,\omega)+d_{\cba}(\omega, \tau)
\ee
holds true.
\end{theorem}

\begin{proof}
If $\omega$ is pure or both $\rho$ and $\tau$ are pure, then by \eqref{eq:qW-dist-expl-pure} and \eqref{eq:self-dist-explicit} we have
\begin{equation}
    d_{\cA}^2(\rho,\omega)=\sum_{j=1}^N\left( \tr_{\cH}[\sqrt{\rho} A_j \sqrt{\rho} A_j ]+\tr_{\cH}[\sqrt{\omega}A_j \sqrt{\omega} A_j ]-2 \tr_{\cH} [\rho A_j]\tr_{\cH}[\omega A_j]\right)
\end{equation}
and
\begin{equation}
    d_{\cA}^2(\omega, \tau )=\sum_{j=1}^N\left( \tr_{\cH}[\sqrt{\omega} A_j \sqrt{\omega} A_j]+\tr_{\cH} [\sqrt{\tau} A_j \sqrt{\tau} A_j]-2\tr_{\cH} [\omega A_j] \tr_{\cH} [\tau A_j]\right).
\end{equation}
By relaxing the infimum in the definition of the quantum Wasserstein divergence to the tensor product coupling, we get
\begin{equation}
    d_{\cA}^2(\rho,\tau) \leq \sum_{j=1}^N\left( \tr_{\cH}[\sqrt{\rho} A_j \sqrt{\rho} A_j]+\tr_{\cH} [\sqrt{\tau} A_j \sqrt{\tau} A_j]-2\tr_{\cH} [\rho A_j] \tr_{\cH} [\tau A_j]\right).
\end{equation}
By \eqref{eq:concavity} the quantum Wasserstein divergence is a non-negative real number and hence eq. \eqref{eq:tri-in} is equivalent to
\begin{equation} \label{eq:tri-in-eq}
2 d_{\cba}(\rho,\omega) d_{\cba}(\omega, \tau) \geq d_{\cba}^2(\rho, \tau)-\left(d_{\cba}^2(\rho,\omega) + d_{\cba}^2(\omega, \tau)\right).
\end{equation}

If $X \in \sh$ is a state having finite energy with respect to the observable $Y \in \cL(\cH)^{sa},$ then $X^{1/4}YX^{1/4}$ is a Hilbert-Schmidt operator. Indeed, let $\sum_{j=1}^{\infty} \lambda_j \ket{\varphi_j}\bra{\varphi_j}$ be the spectral resolution of $X,$ and let us compute the trace of the positive operator $(X^{1/4}YX^{1/4})^2$ by
$$
\tr_{\cH}[(X^{1/4}YX^{1/4})^2]=\sum_{j=1}^{\infty} \bra{\varphi_j}X^{1/4} Y X^{1/4} X^{1/4} Y X^{1/4} \ket{\varphi_j}=
$$
$$
=\sum_{j=1}^{\infty} \norm{X^{1/4}YX^{1/4}\varphi_j}^2=\sum_{j=1}^{\infty} \ler{\lambda_{j}^{1/4}}^2\norm{X^{1/4}Y\varphi_j}^2.
$$
Now let us note that this is precisely the trace of $\sqrt{X}Y\sqrt{X}Y$ that we already have shown to be a trace-class operator. This latter statement can be checked by the direct computation
$$
\tr_{\cH}[\sqrt{X}Y\sqrt{X}Y]=\sum_{j=1}^{\infty} \bra{\varphi_j} \sqrt{X}Y\sqrt{X}Y \ket{\varphi_j}=\sum_{j=1}^{\infty} \bra{\sqrt{X}\varphi_j} Y\sqrt{X}Y \ket{\varphi_j}= \sum_{j=1}^{\infty} \sqrt{\lambda_j} \norm{X^{1/4} Y \varphi_j}^2.
$$
So both $\sqrt{X}$ and $X^{1/4}YX^{1/4}$ are Hilbert-Schmidt operators, and by the Cauchy-Schwarz inequality for the Hilbert-Schmidt inner product of them we get 
$$
\tr_{\cH} [\sqrt{X}Y \sqrt{X}Y]=\tr_{\cH}[(X^{1/4}YX^{1/4})^2] \tr_{\cH}[(\sqrt{X})^2] \geq (\tr_{\cH} [X^{1/4} Y X^{1/4} X^{1/2}])^2=(\tr_{\cH} [X Y])^2,
$$ 
where we used that both $X^{1/4} Y X^{3/4}$ and $X Y$  are trace-class and hence their traces coincide, and hence we get the following upper bound for the right-hand side (RHS) of \eqref{eq:tri-in-eq}:
$$
RHS \leq \sum_{j=1}^N\left( \tr_{\cH}[\sqrt{\rho} A_j \sqrt{\rho} A_j]+\tr_{\cH} [\sqrt{\tau} A_j \sqrt{\tau} A_j]-2\tr_{\cH} [\rho A_j] \tr_{\cH} [\tau A_j]\right)-
$$
$$
-\sum_{j=1}^N\left( \tr_{\cH}[\sqrt{\rho} A_j \sqrt{\rho} A_j ]+\tr_{\cH}[\sqrt{\omega}A_j \sqrt{\omega} A_j ]-2 \tr_{\cH} [\rho A_j]\tr_{\cH}[\omega A_j]\right)-
$$
$$
-\sum_{j=1}^N\left( \tr_{\cH}[\sqrt{\omega} A_j \sqrt{\omega} A_j]+\tr_{\cH} [\sqrt{\tau} A_j \sqrt{\tau} A_j]-2\tr_{\cH} [\omega A_j] \tr_{\cH} [\tau A_j]\right) \leq
$$
$$
\leq \sum_{j=1}^N \ler{-2 \tr_{\cH} [\rho A_j]\tr_{\cH} [\tau A_j]- 2 (\tr_{\cH} [\omega A_j])^2+2\tr_{\cH}[\rho A_j]\tr_{\cH} [\omega A_j]+2 \tr_{\cH} [\omega A_j]\tr_{\cH} [\tau A_j]}=
$$
 
\begin{equation}
    = \sum_{j=1}^N 2 \ler{\tr_{\cH} [\omega A_j] -\tr_{\cH}[\rho A_j]}\ler{\tr_{\cH} [\tau A_j]-\tr_{\cH} [\omega A_j]}.
\end{equation}
Now a Cauchy-Schwartz for the Euclidean space $\mathbb{R}^N$ tells us that
$$
\sum_{j=1}^N 2 \ler{\tr_{\cH} [\omega A_j] -\tr_{\cH}[\rho A_j]}\ler{\tr_{\cH} [\tau A_j]-\tr_{\cH} [\omega A_j]}
\leq
$$
$$
\leq
2 \left( \sum_{j=1}^N (\tr_{\cH} [\omega A_j] -\tr_{\cH} [\rho A_j])^2 \right)^{1/2} \left( \sum_{k=1}^N (\tr_{\cH} [\tau A_k]-\tr_{\cH} [\omega A_k])^2 \right)^{1/2}
\leq
$$
$$
\leq
2
\left( \sum_{j=1}^N\left( \tr_{\cH}[\sqrt{\omega}A_j \sqrt{\omega} A_j ]+\tr_{\cH}[\sqrt{\rho}A_j \sqrt{\rho} A_j]-2 \tr_{\cH} [\omega A_j] \tr_{\cH}[\rho A_j]\right)\right)^{1/2}
\times
$$
$$
\times
\left( \sum_{j=1}^N\left( \tr_{\cH}[\sqrt{\tau}A_j \sqrt{\tau} A_j ]+\tr_{\cH}[\sqrt{\omega}A_j \sqrt{\omega} A_j]-2 \tr_{\cH} [\tau A_j] \tr_{\cH}[\omega A_j]\right)\right)^{1/2}=
$$
$$
=2 d_{\cba}(\omega, \rho) d_{\cba}(\tau,\omega)
$$
where we used again the Cauchy-Schwarz inequality for the Hilbert-Schmidt inner product \\  $\tr_{\cH} [\sqrt{X}Y \sqrt{X}Y] \geq (\tr_{\cH} [X Y])^2,$ for $X=\rho,\tau,\omega$ and $Y=A_j.$ This completes the proof of \eqref{eq:tri-in-eq} and hence that of the Theorem.
\end{proof}

\section{Numerical evidence for the triangle inequality for generic triplets of states}\label{s:numerics}

The quantum optimal transport problem is a semidefinite programming task, and we used Wolfram Mathematica \cite{wolf-math} to perform numerical simulations. The data generated during our experiments (Mathematica notebooks, their pdf images, and the raw data exported from the notebooks) is available online, see \cite{num-data}.
\par
In the quantum bit ($\cH=\C^2$) case, we ran the following experiment. We chose four pairs of random states $(\rho_{(0,1)},\tau_{(0,1)}), (\rho_{(0,2)},\tau_{(0,2)}), (\rho_{(0,3)},\tau_{(0,3)}),(\rho_{(0,4)},\tau_{(0,4)})$ and four triples of random self-adjoint operators: $\cA_{(0,1)},\cA_{(0,2)},\cA_{(0,3)}$ and $\cA_{(0,4)}.$ The random states are normalized \emph{Wishart matrices:} they are of the form $\frac{X^*X}{\tr(X^*X)}$ where $X$ is a $2 \times 2$ random matrix with i.i.d. complex standard Gaussian entries. The random self-adjoint operators are defined similarly: they are of the form $Y+Y^*$ where the elements of $Y$ are i.i.d. complex Gaussians. Then we let $\omega \in \sh$ run on the following lattice within the state space:
\be \label{eq:lat-def}
\mathrm{Lat}(\cS(\C^2)):=\lers{ \fel \ler{I+\frac{1}{10}(j \sigma_1+k \sigma_2+l \sigma_3)}\, \middle| \, j,k,l \in \Z, \, j^2+k^2+l^2 \leq 100}.
\ee
We computed the minimal gap between the two sides of the triangle inequality
\be
\min_{\omega \in \mathrm{Lat}(\cS(\C^2))}\ler{d_{\cA_{(0,m)}}(\rho_{(0,n)},\omega)+d_{\cA_{(0,m)}}(\omega, \tau_{(0,n)})-d_{\cA_{(0,m)}}(\rho_{(0,n)},\tau_{(0,n)})}
\ee
for every $m,n \in {1,2,3,4}.$
We found the following $4 \times 4=16$ elements long list of minimal gaps: 
$$
\left[
\ba{ccccc}
m/n & 1 & 2 & 3 & 4 \\
1 & 0.310819 & 0.528506 &  0.760247 & 0.352543\\
2 & 0.218016 & 0.715538 & 0.590063 & 0.453942 \\
3 & 0.280697 & 0.642319 & 0.669042 & 0.800527 \\
4 & 0.195821 & 0.443850 & 0.447589 & 0.401331
\ea
\right].
$$
Then we turned to the case of higher dimensions $d:=\mathrm{dim}(\cH)=3,4,5.$ In each dimension, we generated $4000$ triples of random states $$(\rho_{(d,n)}, \omega_{(d,n)}, \tau_{(d,n)}),\qquad n=1, \dots, 4000$$ which are i.i.d. normalised $d \times d$ Wishart matrices, and $4000$ triples of random self-adjoint matrices $\cA_{(d,n)} \, (1 \leq n \leq 4000).$ Then we computed the minimal gap
\be \label{eq:dim-min-gap}
\mathrm{mg}(d):=
\min_{n \in \lers{1, \dots, 4000}} \ler{d_{\cA_{(d,n)}}(\rho_{(d,n)},\omega_{(d,n)})+d_{\cA_{(d,n)}}(\omega_{(d,n)}, \tau_{(0,n)})-d_{\cA_{(0,m)}}(\rho_{(d,n)},\tau_{(d,n)})}.
\ee
We found that $\mathrm{mg}(3)=0.854168, \mathrm{mg}(4)=1.89892,$ and $\mathrm{mg}(5)=2.69551.$
\par
That is, we found strong numerical evidence indicating that the conclusion of Theorem \ref{thm:main} holds in full generality, without any additional assumption on the states $\rho, \omega, \tau$ involved.
\par
These minima are convincing but do not tell too much about how the gap depends on the states involved. Therefore, we worked out the following illustrative examples. 
\par
In the first example, we considered $\cH=\C^2$ and chose the deterministic states $$\rho_{(1,1)}=\frac{1}{2}\ler{I+\frac{1}{\sqrt{2}}\sigma_1+\frac{1}{\sqrt{3}}\sigma_2}\qquad\mbox{and}\qquad \tau_{(1,1)}=\frac{1}{2}\ler{I+\frac{1}{3}\sigma_2+\frac{1}{4}\sigma_3}.$$ We singled out the section $z=\frac{1}{\sqrt{2}}$ of the Bloch ball, that is, we took $\omega$'s of the form $$\omega=\frac{1}{2}\ler{I+x \sigma_1+y \sigma_2 +\frac{1}{\sqrt2}\sigma_3}$$ where $x^2+y^2 \leq \frac{1}{2}.$ The set of self-adjoint matrices generating the quadratic cost operator is chosen to be $\cA_{(1,1)}=\lers{\sigma_1, \sigma_3}.$
We considered the gap
$$
d_{\cA_{(1,1)}}(\rho_{(1,1)},\omega_{(1,1)}(x,y))+d_{\cA_{(1,1)}}(\omega_{(1,1)}(x,y), \tau_{(1,1)})-d_{\cA_{(1,1)}}(\rho_{(1,1)},\tau_{(1,1)})
$$
where $\omega_{(1,1)}(x,y)=\frac{1}{2}\ler{I+x \sigma_1+y \sigma_2 +\frac{1}{\sqrt2}\sigma_3},$ and plotted it as a function of $x$ and $y$ in Figure \ref{fig:c2-det}.

\begin{figure}[H]
\caption{The plot of the gap in the first scenario: $\cH=\C^2,$ the states and the cost are chosen to be nice}
\label{fig:c2-det}
\includegraphics[width=0.93\textwidth]{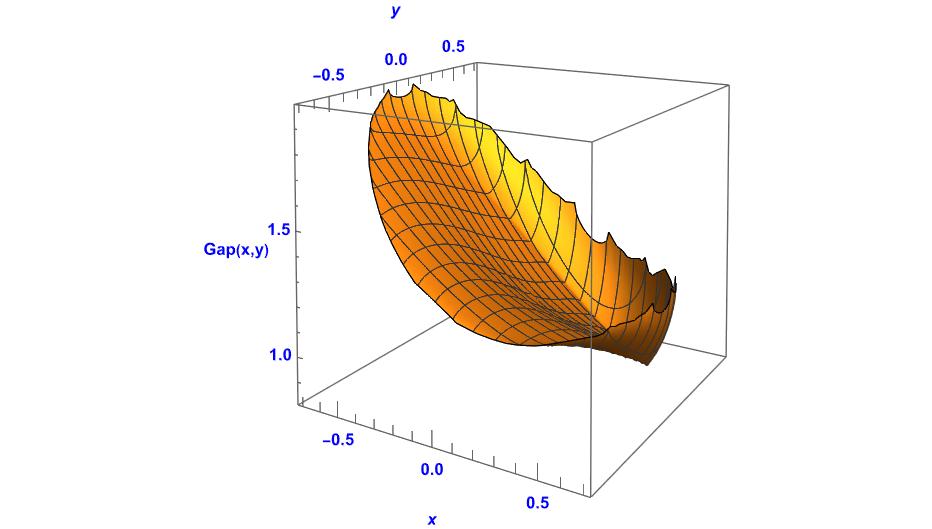}
\end{figure}

The second example deals with $\cH=\C^4.$ We chose
$$
\rho_{(1,2)}=\frac{1}{4}(I+\frac{1}{10} \sigma_1 \otimes \sigma_1+ \frac{1}{5}\sigma_2 \otimes \sigma_0 + \frac{3}{10} \sigma_3 \otimes \sigma_0)
$$ 
and
$$
\tau_{(1,2)}=\frac{1}{4}(I+\frac{3}{10} \sigma_0 \otimes \sigma_3+ \frac{1}{5}\sigma_1 \otimes \sigma_3 + \frac{1}{10} \sigma_3 \otimes \sigma_0).
$$
We considered $\omega$'s of the form
$$\omega_{(1,2)}(x,y)=$$
\be \label{eq:omega-reg-def}
=\frac{1}{4}(I+x \sigma_0 \otimes \sigma_1 + y \sigma_0 \otimes \sigma_2 +\frac{1}{10} \sigma_1 \otimes \sigma_0 + \frac{1}{10} \sigma_1 \otimes \sigma_1+ \frac{1}{10} \sigma_1 \otimes \sigma_2+\frac{3}{10} \sigma_2 \otimes \sigma_0+\frac{1}{5} \sigma_2 \otimes \sigma_2)
\ee
and the cost governed by the collection of all possible tensor products of Pauli matrices:
$$
\cA_{(1,2)}=\lers{\sigma_j \otimes \sigma_k \, \middle| \, j,k \in {0,1,2,3}, \, (j,k) \neq (0,0)}.
$$
Figure \ref{fig:c4-det} shows the plot of the gap
$$
d_{\cA_{(1,2)}}(\rho_{(1,2)},\omega_{(1,2)}(x,y))+d_{\cA_{(1,2)}}(\omega_{(1,2)}(x,y), \tau_{(1,2)})-d_{\cA_{(1,2)}}(\rho_{(1,2)},\tau_{(1,2)})
$$

\begin{figure}[H]
\caption{The plot of the gap in the second scenario: $\cH=\C^4,$ the states and the cost are chosen to be nice}
\label{fig:c4-det}
\includegraphics[width=0.60\textwidth]{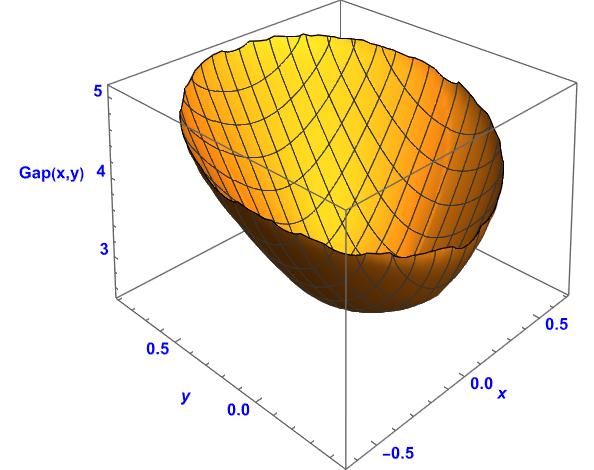}
\end{figure}

The following Figure \ref{fig:c2-rand} concerns the qubit case again, but now the states $\rho_{(1,3)}$ and $\tau_{(1,3)}$ are random (independent normalized Wishart) densities, the quadratic cost operator is generated by random self-adjoint operators, and $\omega_{(1,3)}=\omega_{(1,3)}(x,y)$ runs over the section $z=\frac{1}{5},$ that is, $\omega_{(1,3)}(x,y)=\frac{1}{2}\ler{I+x \sigma_1+y \sigma_2 +\frac{1}{5}\sigma_3}$ where $x^2+y^2 \leq \frac{24}{25}.$

\begin{figure}[H]
\caption{The plot of the gap in the third scenario: $\cH=\C^2,$ the states and the cost are random}
\label{fig:c2-rand}
\includegraphics[width=0.5\textwidth]{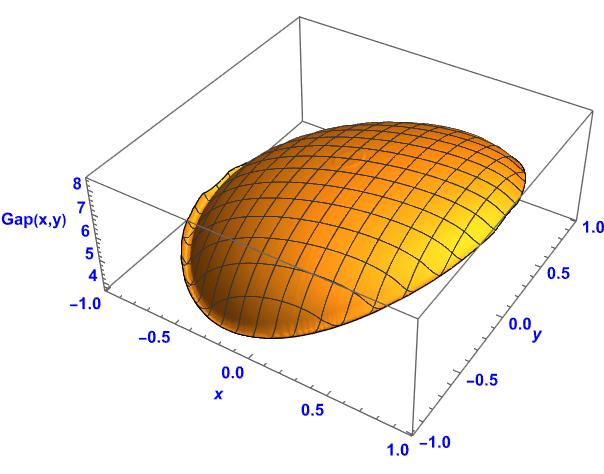}
\end{figure}

Finally, Figure \ref{fig:c4-rand} shows the behaviour of the gap in the $\cH=\C^4$ case, with random states $\rho_{(1,4)}$ and $\tau_{(1,4)},$ and random quadratic cost. The third state $\omega_{(1,4)}$ runs over the region described in \eqref{eq:omega-reg-def}.

\begin{figure}[H]
\caption{The plot of the gap in the second scenario: $\cH=\C^4,$ the states and the cost are random}
\label{fig:c4-rand}
\includegraphics[width=0.7\textwidth]{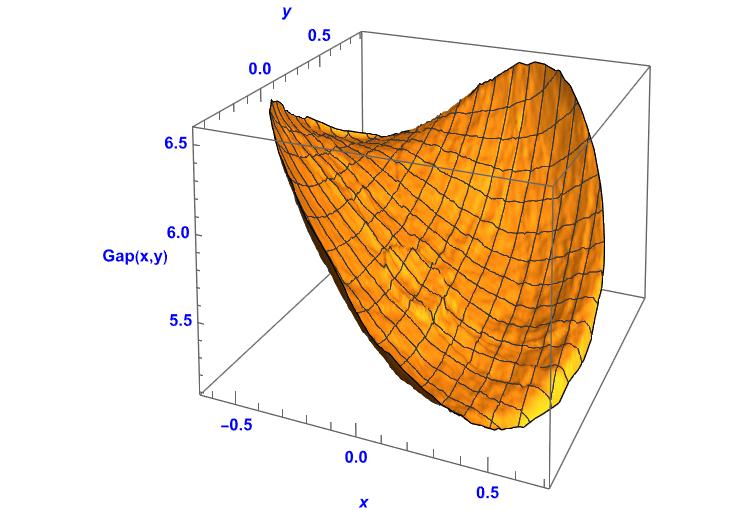}
\end{figure}

\section{Steps towards an analytic proof of the triangle inequality in full generality}\label{s:anal-spec}

The positive numerical results presented in Section \ref{s:numerics} encouraged us to take steps towards an analytical proof of the general case. In Subsection \ref{ss:qub-sym} we study the special case when $\cH=\C^2$ and the cost operator $C$ is the \emph{symmetric cost}. As we will see, it is vital to get useful lower bounds for quantum Wasserstein distances. One way to get these lower bounds is to find (not necessarily optimal) solutions of the dual optimization problem --- see Proposition \ref{prop:dual}. Another way of getting lower bounds is described in Subsection \ref{ss:low-est}. The starting point of this method is the observation that the summands of the quadratic cost operators are gaps between the arithmetic and geometric means of certain operators.

\subsection{Triangle inequality for qubits --- the case of the symmetric cost}\label{ss:qub-sym}

In the special case $\cH=\C^2,$ elements of $\mathcal{S}(\cH)$ can be represented by vectors using the Bloch representation. The \emph{Bloch vector} $\bb_{\rho}$ of a state $\rho \in \mathcal{S}(\cH)$ is defined by
$$
\R^3 \ni \bb_{\rho}:=\ler{\tr_{\cH}(\rho \sigma_j) }_{j=1}^3
$$
where the $\sigma_j$'s are the Pauli operators
\be \label{eq:pauli}
\sigma_1=\lesq{\ba{cc} 0 & 1  \\  1 & 0  \ea}\qquad \sigma_2=\lesq{\ba{cc} 0 & -i  \\  i & 0  \ea}\qquad \sigma_3=\lesq{\ba{cc} 1 & 0  \\  0 & -1  \ea}.
\ee
The positivity condition $\rho \geq 0$ is equivalent to the Euclidean length of $\bb_{\rho}$ being at most $1.$

In this section we stick to $\cH=\C^2$ and we choose the cost operator to be
\be \label{eq:symm-cost-def}
C_{s}=\sum_{j=1}^3 \ler{\sigma_j \otimes I_{\C^2}^T-I_{\C^2} \otimes \sigma_j^T}^2
\ee
where the ``s'' in the subscript of $C_s$ refers to the symmetry of $C_s$ meaning that it involves all the Pauli matrices. The induced quantum Wasserstein distance \eqref{eq:2-Wass-quant-def} is denoted by $D_s,$ and the corresponding modified quantum Wasserstein distance, or \emph{Wasserstein divergence} \eqref{eq:mod-qw-dist-def} is denoted by $d_s.$ 
\par
We aim to prove the triangle inequality
\be \label{eq:tri-ineq-spec}
d_s(\rho, \tau) \leq d_s(\rho, \omega)+d_s(\omega, \tau)
\ee
for a reasonably large class of states $\rho, \omega, \tau \in \cS(\C^2).$
\par
As we have seen, \eqref{eq:tri-ineq-spec} is equivalent to
\be \label{eq:tri-ineq-spec-equiv}
2 d_s(\rho,\omega) d_s(\omega, \tau) \geq d_s^2(\rho, \tau)-\left(d_s^2(\rho,\omega) + d_s^2(\omega, \tau)\right).
\ee
By the very definition of the Wasserstein divergence \eqref{eq:mod-qw-dist-def}, the above \eqref{eq:tri-ineq-spec-equiv} can be written as
$$
2\Big(D_s^2(\rho,\omega)-\fel \big(D_s^2(\rho, \rho) + D_s^2(\omega, \omega)\big)\Big)^{\fel}\Big(D_s^2(\omega,\tau)-\fel \big(D_s^2(\omega, \omega) + D_s^2(\tau, \tau)\big)\Big)^{\fel} \geq
$$
$$
\geq
D_s^2(\rho,\tau)-\fel \ler{D_s^2(\rho, \rho) + D_s^2(\tau, \tau)}-D_s^2(\rho,\omega)+\fel \ler{D_s^2(\rho, \rho) + D_s^2(\omega, \omega)} -
$$
\be \label{eq:lhs-rhs}
-D_s^2(\omega,\tau)+\fel \ler{D_s^2(\omega, \omega) + D_s^2(\tau, \tau)}
=D_s^2(\rho,\tau)-D_s^2(\rho,\omega)-D_s^2(\omega,\tau)+ D_s^2(\omega, \omega).
\ee
Our strategy to prove \eqref{eq:lhs-rhs} is to give a lower bound $\mathrm{LB}$ of the left-hand side of \eqref{eq:lhs-rhs} and an upper bound $\mathrm{UB}$ of the right side of \eqref{eq:lhs-rhs} such that $\mathrm{LB} \geq \mathrm{UB}.$
The following statement will turn out to be useful in deriving such lower and upper bounds.
\begin{proposition} \label{prop:dual}
We have the lower bound
\be \label{eq:qwd-lower}
D_s^2(\rho, \omega) \geq 4 \abs{\bb_\rho-\bb_\omega}_2
\ee
for any $\rho, \omega \in \cS(\C^2)$ where $\bb_\rho$ is the \emph{Bloch vector} of $\rho$ and $\abs{\cdot}_2$ denotes the $l_2$ (Euclidean) norm.
\end{proposition}

\begin{proof}
The first step is to prove that
\be \label{eq:loewner-bound}
C_s \geq X \otimes I_{\C^2}^{T}-I_{\C^2} \otimes X^T 
\ee
for any $X \in \mathcal{L}^{sa}(\C^2)$ satisfying $-4 I_{\C^2} \leq X \leq 4 I_{\C^2}.$
A crucial observation of \cite{GPTV23} is that the symmetric cost operator $C_s$ defined in \eqref{eq:symm-cost-def} is \emph{unitary invariant,} that is,
\be \label{eq:sym-cost-invariance}
\ler{U \otimes \ler{U^T}^*} C_{s} \ler{U^* \otimes U^T}=C_{s}
\ee
for any $U \in \mathbf{U}(2).$ Therefore, the spectral resolution of $C_s,$ which is computed in \cite{GPTV23} in the computational basis, is valid in any orthonormal basis.
\par
Let $\{e_1,e_2\} \subset \C^2$ be the eigenbasis of $X$ and let $\{f_1,f_2\} \subset (\C^2)^*$ be the corresponding dual basis.
Then we have 
\be \label{eq:sym-cost-spectral}
C_S=0 \cdot \ket{v_0}\bra{v_0}+8\cdot(\ket{v_1}\bra{v_1}+\ket{v_2}\bra{v_2}+\ket{v_3}\bra{v_3})
\ee
where
$$
v_0=\frac{1}{\sqrt{2}}(e_1 \otimes f_1 +e_2 \otimes f_2), \quad
v_1=\frac{1}{\sqrt{2}}(e_1 \otimes f_2 +e_2 \otimes f_1),
$$
$$
v_2=\frac{1}{\sqrt{2}}(-i e_1 \otimes f_2 + i e_2 \otimes f_1), \quad
v_3=\frac{1}{\sqrt{2}}(e_1 \otimes f_1 - e_2 \otimes f_2).
$$
Now let us note that $X \otimes I_{\C^2}^{T}-I_{\C^2} \otimes X^T$ always annihilates $v_0=\frac{1}{\sqrt{2}}(e_1 \otimes f_1 +e_2 \otimes f_2).$ Indeed,
$$
(X \otimes I_{\C^2}^{T}-I_{\C^2} \otimes X^T)(e_1 \otimes f_1 +e_2 \otimes f_2)=
$$
$$
=(\lambda_1 e_1) \otimes f_1 + (\lambda_2 e_2) \otimes f_2- e_1 \otimes (\lambda_1 f_1) - e_2 \otimes (\lambda_2 f_2)=0.
$$
And on the subspace orthogonal to $v_0$ it is enough to guarantee that 
$$
X \otimes I_{\C^2}^{T}-I_{\C^2} \otimes X^T \leq 8 I
$$
which is clear as both $X \otimes I_{\C^2}^{T}$ and $-I_{\C^2} \otimes X^T$ is bounded from above by $4 I.$
So we have proved \eqref{eq:loewner-bound} and now we use it to get useful lower bounds on quantum Wasserstein distances by appropriate choices of $X.$
The lower bound on the symetric cost \eqref{eq:loewner-bound} implies that for any $\rho, \omega \in \cS(\C^2)$ we have
$$
D_s^2(\rho, \omega)
=\tr_{\hohc}\ler{\Pi_0 C_s}\geq \tr_{\hohc}\ler{\Pi_0 (X \otimes I_{\C^2}^{T}-I_{\C^2} \otimes X^T)}=
$$
$$
=\tr_{\cH}(\omega X)-\tr_{\cH}(\rho X)=\tr_{\cH}((\omega-\rho) X).
$$
We aim for the highest possible lower bound on $D_s^2(\rho, \omega),$ so let us choose $X$ as follows:
\be \label{eq:x-def}
X:=4 \frac{\bb_{\omega}-\bb_{\rho}}{\abs{\bb_{\omega}-\bb_{\rho}}_2} \cdot \mathbb{\sigma}
\ee
This choice satisfies the condition $-4 I_{\C^2} \leq X \leq 4 I_{\C^2}$ and
$$
\tr_{\cH}((\omega-\rho) X)
=\tr_{\cH}\ler{
\ler{\frac{1}{2}(\bb_{\omega}-\bb_{\rho}) \cdot \mathbb{\sigma}}
\ler{4 \frac{\bb_{\omega}-\bb_{\rho}}{\abs{\bb_{\omega}-\bb_{\rho}}_2} \cdot \mathbb{\sigma}}
}=
$$
$$
=\frac{2}{\abs{\bb_{\omega}-\bb_{\rho}}_2}2\abs{\bb_{\omega}-\bb_{\rho}}_2^2=4 \abs{\bb_{\omega}-\bb_{\rho}}_2
$$
which is precisely the lower bound appearing in \eqref{eq:qwd-lower}.
\end{proof}
Moreover, we have the following explicit formula for the self-distance $D_s^2(\rho,\rho).$
\begin{proposition} \label{prop:self-dist}
We have 
\be \label{eq:self-dist-expl}
D_s^2(\rho,\rho)=4\ler{1-\sqrt{1-\abs{\bb_\rho}_2^2}}
\ee
where $\abs{\cdot}_2$ denotes the Euclidean ($l_2$) norm.
\end{proposition}

\begin{proof}
The purification of $\rho$ realizes the self-distance:
$D_s^2(\rho,\rho)=\Bra{\sqrt{\rho}} C_s \Ket{\sqrt{\rho}}.$
The spectral decomposition of $C_s$ is 
$$
C_s=8
\ler{\frac{1}{2}\Ket{\sigma_1}\Bra{\sigma_1}+\frac{1}{2}\Ket{\sigma_2}\Bra{\sigma_2}+\frac{1}{2}\Ket{\sigma_3}\Bra{\sigma_3}}
$$
and it follows from the spectral resolution of $\rho$ that
$$
\sqrt{\rho}
=\sqrt{\lambda} \fel\ler{I+\frac{\bb_{\rho}}{\abs{\bb_{\rho}}_2} \cdot \mathbb{\sigma}}
+\sqrt{1-\lambda} \fel\ler{I-\frac{\bb_{\rho}}{\abs{\bb_{\rho}}_2} \cdot \mathbb{\sigma}}
$$
where $\lambda=\fel(1+\abs{\bb_\rho}_2).$
Therefore,
$$
\Bra{\sqrt{\rho}} C_s \Ket{\sqrt{\rho}}=8 \norm{\fel (\sqrt{\lambda}-\sqrt{1-\lambda}) \frac{\bb_{\rho}}{\abs{\bb_{\rho}}_2} \cdot \mathbb{\sigma}}_{HS}^2=
$$
$$
=2\ler{1-2\sqrt{\fel(1+\abs{\bb_\rho}_2)\fel(1-\abs{\bb_\rho}_2)}} \frac{2 \abs{\bb_\rho}_2^2}{\abs{\bb_\rho}_2^2}=4\ler{1-\sqrt{1-\abs{\bb_\rho}_2^2}}.
$$
\end{proof}

With \eqref{eq:qwd-lower} and \eqref{eq:self-dist-expl} in hand, we can give the following upper bound on the right-hand side of \eqref{eq:lhs-rhs}:
$$
D_s^2(\rho,\tau)-D_s^2(\rho,\omega)-D_s^2(\omega,\tau)+ D_s^2(\omega, \omega)
\leq 
$$
$$
\leq 6-2 \bb_{\rho} \cdot \bb_{\tau}-4\abs{\bb_{\rho}-\bb_{\omega}}_2-4\abs{\bb_{\omega}-\bb_{\tau}}_2+4\ler{1-\sqrt{1-\abs{\bb_\omega}_2^2}}
$$
where $D_s^2(\rho,\tau)$ was bounded by the cost of the independent coupling.
\par
The lower bound on the left hand side of \eqref{eq:lhs-rhs} that we can obtain by Proposition \ref{prop:dual} and Proposition \ref{prop:self-dist} reads as follows:
$$
2\ler{D_s^2(\rho,\omega)-\fel \ler{D_s^2(\rho, \rho) + D_s^2(\omega, \omega)}}^{\fel}\ler{D_s^2(\omega,\tau)-\fel \ler{D_s^2(\omega, \omega) + D_s^2(\tau, \tau)}}^{\fel} \geq
$$
$$
\geq
2\ler{4\abs{\bb_{\rho}-\bb_{\omega}}_2-2\ler{1-\sqrt{1-\abs{\bb_\rho}_2^2}} -2\ler{1-\sqrt{1-\abs{\bb_\omega}_2^2}}}^{\fel} \times
$$
$$
\times 
\ler{4\abs{\bb_{\omega}-\bb_{\tau}}_2-2\ler{1-\sqrt{1-\abs{\bb_\omega}_2^2}}-2\ler{1-\sqrt{1-\abs{\bb_\tau}_2^2}}}^{\fel}.
$$
We summarize the above computations in the following corollary.

\begin{corollary} \label{cor:sym-suff}
Let us choose $\rho, \omega, \tau \in \cS(\C^2)$ such that their Bloch vectors satisfy
$$
6-2 \bb_{\rho} \cdot \bb_{\tau}-4\abs{\bb_{\rho}-\bb_{\omega}}_2-4\abs{\bb_{\omega}-\bb_{\tau}}_2+4\ler{1-\sqrt{1-\abs{\bb_\omega}_2^2}} \leq
$$
$$
\leq
2\ler{4\abs{\bb_{\rho}-\bb_{\omega}}_2-2\ler{1-\sqrt{1-\abs{\bb_\rho}_2^2}} -2\ler{1-\sqrt{1-\abs{\bb_\omega}_2^2}}}^{\fel} \times
$$
\be \label{eq:sym-suff-cond}
\times 
\ler{4\abs{\bb_{\omega}-\bb_{\tau}}_2-2\ler{1-\sqrt{1-\abs{\bb_\omega}_2^2}}-2\ler{1-\sqrt{1-\abs{\bb_\tau}_2^2}}}^{\fel}.
\ee
Then the quantum Wasserstein divergence corresponding to the symmetric cost operator satisfies the triangle inequality
$$
d_s(\rho, \tau) \leq d_s(\rho, \omega)+d_s(\omega, \tau).
$$
\end{corollary}

\begin{remark}
There are various ways to obtain easy-to-check examples of states $\rho, \omega, \tau, \in \cS(\C^2)$ that satisfy the assumption of Corollary \ref{cor:sym-suff}. One way is to fix $\omega=\frac{1}{2}I,$ that is, $\bb_{\omega}=0.$ In this case, simple $2$-variable calculus shows that \eqref{eq:sym-suff-cond} is satisfied whenever both $\abs{\bb_{\rho}}_2$ and $\abs{\bb_{\tau}}_2$ are at least $\fel,$ no matter what the angle between $\bb_{\rho}$ and $\bb_{\tau}$ is. Moreover, a numerical test shows if we choose a random triplet of states (the states are chosen independently according to the uniform distribution on the Bloch ball), it will satisfy \eqref{eq:sym-suff-cond} with high probability (around $85\%$). The details of this numerics is presented in the subfolder "Corollary-4" of \cite{num-data}. Note, however, that there are states that do not satisfy \eqref{eq:sym-suff-cond}: the easiest example is $\rho=\omega=\tau=\fel I.$
\end{remark}

\begin{remark}
Note that the estimate \eqref{eq:loewner-bound} may be sharp:
the numerics tells us that it gives the precise transport cost in the special cases
\begin{itemize}
\item $\rho=\fel(I+\fel \sigma_j), \, \omega=\fel(I+\fel \sigma_k), \quad j \neq k$
\item $\rho=\fel(I+\alpha \sigma_j), \, \omega=\fel(I+\beta \sigma_j),  \quad \mathrm{sgn}(\alpha)=-\mathrm{sgn}(\beta).$
\end{itemize}
\end{remark}

\subsection{Lower estimation of the cost function}\label{ss:low-est}
Assuming that $A_j \geq 0$ for every $j \in \{1, \dots, k\},$ we can write the cost operator as the difference of an arithmetic and a geometric mean as follows
$$
C=\sum_{j=1}^k \ler{A_j \otimes I_{\cH^*} - I_{\cH} \otimes A_j^T}^2=
$$
$$
=\sum_{j=1}^k \ler{A_j^2 \otimes I_{\cH^*}+I_{\cH}\otimes (A_j^T)^2 - 2A_j\otimes A_j^T}=
$$
$$
=2\sum_{j=1}^k \ler{\frac{A_j^2 \otimes I_{\cH^*}+I_{\cH}\otimes (A_j^T)^2}{2}-[(A_j^2\otimes I_{\cH^*})(I_{\cH}\otimes (A_j^T)^2)]^{1/2}}.
$$
Introducing the function
$$
f(t)=\frac{1+t}{2}-\sqrt{t},\quad (t\ge 0),
$$
the cost operator can be written as
$$
C=2\sum_{j=1}^k(A_j^2\otimes I_{\cH^*})f(A_j^{-2}\otimes (A_j^T)^2).
$$
$f$ is an operator convex function on $[0,\infty)$ and its tangent line at $s$ is given by
$$
g_s(t)=\frac{1-\sqrt{s}}{2}+\frac{\sqrt{s}-1}{2\sqrt{s}}t.
$$
By the convexity of $f$ we have
$$f(t)\ge g_s(t),\quad (s>0, t\ge 0),$$
and
$$
C=2\sum_{j=1}^k(A_j^2\otimes I_{\cH^*})f(A_j^{-2}\otimes (A_j^T)^2)\ge
$$
$$
\geq
2\sum_{j=1}^k(A_j^2\otimes I_{\cH^*})g_s(A_j^{-2}\otimes (A_j^T)^2)=
$$
$$
=\sum_{j=1}^k\ler{(1-\sqrt{s})A_j^2\otimes I_{\cH^*}+\frac{\sqrt{s}-1}{\sqrt{s}}I_{\cH}\otimes (A_j^T)^2}=
$$
$$
=(1-\sqrt{s})\ler{\sum_{j=1}^k A_j^2}\otimes I_{\cH^*}+\frac{\sqrt{s}-1}{\sqrt{s}}I_{\cH}\otimes \ler{\sum_{j=1}^k(A_j^T)^2}.
$$

We can use this estimation to give a lower bound for the cost function in the state $\Pi$.
$$\tr_{\hohc}[\Pi\, C]\ge \tr_{\hohc}\lesq{\Pi\ler{(1-\sqrt{s})\ler{\sum_{j=1}^k A_j^2}\otimes I_{\cH^*}+\frac{\sqrt{s}-1}{\sqrt{s}}I_{\cH}\otimes \ler{\sum_{j=1}^k(A_j^T)^2}}}=$$
$$
=(1-\sqrt{s})\ler{\tr_{\cH}\lesq{\ler{\sum_{j=1}^k A_j^2}\omega}-\frac{1}{\sqrt{s}}\tr_{\cH^*}\lesq{\ler{\sum_{j=1}^k(A_j^T)^2}\rho^T}}=
$$
$$
=(1-\sqrt{s})\ler{\alpha(\omega)-\frac{1}{\sqrt{s}}\beta(\rho)}
$$
for all $s>0$, where $\alpha(\omega)=\tr_{\cH}\lesq{\ler{\sum_{j=1}^k A_j^2}\omega}$ and 
$$
\beta(\rho)=\tr_{\cH^*}\lesq{\ler{\sum_{j=1}^k(A_j^T)^2}\rho^T}=\tr_{\cH}\lesq{\ler{\sum_{j=1}^k A_j^2}\rho}.
$$
The function
$$
h(s)=(1-\sqrt{s})\ler{\alpha(\omega)-\frac{1}{\sqrt{s}}\beta(\rho)}
$$
has maximum at $s=\frac{\beta(\rho)}{\alpha(\omega)}$. Since
$$
h\ler{\frac{\beta(\rho)}{\alpha(\omega)}}=2\ler{\frac{\alpha(\omega)+\beta(\rho)}{2}-\sqrt{\alpha(\omega)\beta(\rho)}},
$$
we deduce that the cost function can be estimated from below by the Hellinger distance of $\alpha(\omega)$ and $\beta(\rho),$ that is,
$$
\tr_{\hohc}(\Pi\, C)\ge 2\ler{\frac{\alpha(\omega)+\beta(\rho)}{2}-\sqrt{\alpha(\omega)\beta(\rho)}}.
$$

\section{Applications}\label{ss:pros-appl}

The aim of this section is to highlight a direct application of our main result in quantum complexity theory (Subsection \ref{ss:wass-compl}), and to present a brief description of the role of quantum optimal transport distances in mathematical physics (Subsections \ref{ss:mean-field-approx} and \ref{ss:luttinger}). 

\subsection{Wasserstein complexity} \label{ss:wass-compl}

We may define the Wasserstein complexity of quantum channels relying on quantum Wasserstein divergences. Given a finite collections of observables $\cA=\lers{A_1, \dots, A_k}$, let us define the Wasserstein complexity of a channel $\Phi: \, \sh \rightarrow \sh$ by 

\be \label{eq:qw-compl-def}
C_{W}(\Phi):=\max_{\rho \in \sh} d_{\cA}(\rho, \Phi(\rho)). 
\ee

In \cite{liWassersteinComplexityQuantum2022} the quantum Wasserstein complexity is introduced in a similar way, using the quantum generalization of the Hamming-Wasserstein classical metric. A connection between this complexity and the circuit cost of unitary channels is provided, the latter being bounded from below by the former. The argument of this bound relies heavily on the triangle inequality of the Wasserstein distance used. We suggest now that \eqref{eq:qw-compl-def} may have a similar application in light of the triangle inequality proposed in this paper.

To justify our suggestion we now prove a few useful properties of \eqref{eq:qw-compl-def} that is desirable of a complexity quantity.
It follows from the definition, the positive definiteness of $d_{\cA}$ and the lack of self-distance in terms of $d_{\cA}$ that the Wasserstein complexity is faithful: $C_{W}(\Phi)=0$ if and only if $\Phi$ is the identity. 
The Wasserstein complexity is subadditive under concatenation: $C_W(\Phi_2 \circ \Phi_1) \leq C_W(\Phi_2)+C_W(\Phi_1)$. Indeed, 

$$
C_{W}(\Phi_1\circ\Phi_2)=\max_{\rho \in \sh} d_{\cA}(\rho, \Phi_1\circ\Phi_2(\rho))\leq \max_{\rho \in \sh}\left[ d_{\cA}(\rho, \Phi_2(\rho))+d_{\cA}(\Phi_2(\rho), \Phi_1\circ\Phi_2(\rho))\right]
$$

$$
\leq\max_{\rho \in \sh}d_{\cA}(\rho, \Phi_2(\rho))+\max_{\rho \in \sh}d_{\cA}(\Phi_2(\rho), \Phi_1\circ\Phi_2(\rho))
$$

\be
\leq\max_{\rho \in \sh}d_{\cA}(\rho, \Phi_2(\rho))+\max_{\rho \in \sh}d_{\cA}(\rho, \Phi_1(\rho))=C_{W}(\Phi_2)+C_{W}(\Phi_1),
\ee
where the first inequality follows from the triangle inequality for the Wassertein divergence $d_{\cA}$, the second inequality is due maximizing terms of a sum separately, and the last inequality follows from broadening the domain of the second maximum. From this it also follows directly that the Wasserstein complexity is subadditive under tensor products in the following sense: 
$$
C_W(\Phi_1 \otimes \Phi_2) \leq C_W (\Phi_1 \otimes I)+C_W (I \otimes \Phi_2).
$$ 
We note that whether the Wasserstein complexity \eqref{eq:qw-compl-def} is convex or not is an open question.

\subsection{Mean-field approximations of evolution equations} \label{ss:mean-field-approx}
Among the many applications of classical optimal transport, one is particularly important for the development of quantum optimal transport theory. In \cite{Dobrushin}, Dobrushin employed a special approximation based on a transport-related metric, the so-called Kantorovich-Rubinstein metric, to prove the uniqueness of the solution to the Vlasov equation. The Vlasov equations, which describe the limiting situation of weakly interacting particles with a large radius of interaction, are among the most frequently used kinetic equations in statistical mechanics. More recently, Golse, Paul, and Mouhot have extended Dobrushin's approach to the quantum setting. Since Dobrushin used an optimal transport metric to compare $N$-particle densities and their mean-field limits, it was a natural idea to define a transport-related quantity for the purpose of comparing quantum states \cite{GolseMouhotPaul}. The significance of the quantum counterpart of optimal transport (and the Wasserstein metric) in addressing various problems in quantum dynamics became even more evident later through their subsequent papers, see e.g.  \cite{CagliotiGolsePaul,CagliotiGolsePaul-towardsqot,GolseMouhotPaul,GolsePaul-Schrodinger,GolsePaul-lowregularity,GolseTPaul-pseudometrics}.

\subsection{The Luttinger model} \label{ss:luttinger}
The non-equilibrium dymamics of the Luttinger model, describing the low energy physics in Luttinger liquid is a field studied extensively.
The question of how much the time evolved state described by $\rho(t)$ differs from the initial state $\rho(0)$ is often investigated. The paper \cite{DoraBacsi}
study the time evolution of Uhlmann fidelity (or the Loschmidt echo) which measures the overlap between the time evolved and the initial thermal equilibrium states,
is evaluated for arbitrary initial temperatures and quench protocols. Instead of fidelity we can consider other dissimilarity measures between states to describe these effects.
With zero initial temperature, i.e. the initial state is the pure ground state of the Hamiltonian, the Wasserstein distance is easiliy computable and can be used for further
investigations.

\subsection*{Acknowledgment}
We are grateful to the anonymous referees for their insightful suggestions and comments.

\section{Appendix: the code used for the numerical study}

Here we present the \emph{Wolfram Mathematica} \cite{wolf-math} code that we used to obtain the numerical results discussed in Section \ref{s:numerics}.\\

\begin{verbatim}
    RandSelfadjMatrix[l_] := Module[{A, i},
A = RandomVariate[NormalDistribution[], {l, l}] +
I * RandomVariate[NormalDistribution[], {l, l}];
A = A + ConjugateTranspose[A];
A = Chop[A];
A]

RandPositive[l_, r_] := Module[{A},
A = RandomVariate[NormalDistribution[], {l, r}] +
I * RandomVariate[NormalDistribution[], {l, r}];
A = A.ConjugateTranspose[A];
A]

RandState[l_, r_] := Module[{A},
A = RandPositive[l, r];
A = A / Tr[A];
A]

CostfromObservable[Observables_, transpose_] := Sum[
MatrixPower[KroneckerProduct[Observables[[k]],
IdentityMatrix[Length[Observables[[1]]]]] -
KroneckerProduct[IdentityMatrix[Length[Observables[[1]]]], 
If[transpose,Transpose[Observables[[k]]], Observables[[k]]]], 2],
{k, 1, Length[Observables]}]

   SaMCNB[dim_] :=
SaMCNB[dim] = Flatten[{{Table[SparseArray[{{k, k} -> 1,
{dim, dim} -> 0}], {k, 1, dim}]},
Table[Table[SparseArray[{{k, m} -> Sqrt[2] / 2,
{m, k} -> Sqrt[2] / 2, {dim, dim} -> 0}],
{m, k + 1, dim}], {k, 1, dim - 1}], 
Table[Table[SparseArray[{{k, m} -> -I * Sqrt[2] / 2,
{m, k} -> I * Sqrt[2] / 2, {dim, dim} -> 0}],
{m, k + 1, dim}], {k, 1, dim - 1}]}, 2]; 

SaMPB[2] = {{{1, 0}, {0, 1}}, {{0, 1}, {1, 0}},
{{0, -I}, {I, 0}}, {{1, 0}, {0, -1}}};
SaMPB[dim_] := SaMPB[dim] = Flatten[Table[
KroneckerProduct[SaMPB[dim / 2][[j]], SaMPB[2][[k]]],
{j, 1, dim^2 / 4}, {k, 1, 4}], 1];

SaMPNB[dim_] := SaMPNB[dim] = SaMPB[dim] / Sqrt[dim];

QOT[rho_, omega_, C_, dual_, transpose_] 
:= Module[{dim, sol, Pi, x, y, X, Y},
dim = Length[rho];
If[
dual, {x = Table[Symbol["x" <> ToString[n]], {n, dim^2}];
y = Table[Symbol["y" <> ToString[n]], {n, dim^2}];
sol = SemidefiniteOptimization[-Tr[(y.SaMCNB[dim]).omega 
+ (x.SaMCNB[dim]).rho],
{VectorGreaterEqual[{C - KroneckerProduct[(y.SaMCNB[dim]),
IdentityMatrix[dim]] -
KroneckerProduct[IdentityMatrix[dim], If[transpose, Transpose[
(x.SaMCNB[dim])], x.SaMCNB[dim]]], 0}, {"SemidefiniteCone", dim^2}]},
Flatten[{x, y}]];
X = Chop[sol[[1 ;; dim^2, 2]].SaMCNB[dim], 10^-3];
Y = Chop[sol[[dim^2 + 1 ;; 2 * dim^2, 2]].SaMCNB[dim], 10^-3];
Chop[Sqrt[Tr[X.omega + Y.rho]], 10^-3],
X // MatrixForm, Y // MatrixForm},
{x = Table[Symbol["x" <> ToString[n]], {n, dim^4}];
sol = SemidefiniteOptimization[
Chop[Simplify[Tr[(x.SaMCNB[dim^2]).C]], 10^-3],
{ResourceFunction["MatrixPartialTrace"]
[x.SaMCNB[dim^2], 2, {dim, dim}] -> omega,
ResourceFunction["MatrixPartialTrace"]
[x.SaMCNB[dim^2], 1, {dim, dim}] ->
If[transpose, Transpose[rho], rho],
VectorGreaterEqual[{x.SaMCNB[dim^2], 0}, 
{"SemidefiniteCone", dim^2}] },x];
Pi = Chop[Sum[sol[[n, 2]]SaMCNB[dim^2][[n]], {n, 1, dim^4}], 10^-3];
Chop[Sqrt[Tr[C.Pi]], 10^-3], Pi // MatrixForm}]]

ModQOT [rho_, omega_, cost_, dual_, transpose_] 
:= Sqrt[QOT[rho, omega, cost, dual, transpose][[1]]^2 -
(QOT[rho, rho, cost, dual, transpose][[1]]^2 
+ QOT[omega, omega, cost, dual, transpose][[1]]^2) / 2]

TriIneq[rho_, omega_, tau_, C_, dual_, transpose_] 
:= ModQOT[rho, omega, C, dual, transpose] +
ModQOT[omega, tau, C, dual, transpose] 
- ModQOT[rho, tau, C, dual, transpose]

\end{verbatim}

\begin{small}
\bibliographystyle{plainurl}  
\bibliography{references.bib}
\end{small}

\end{document}